\documentclass{stacs_proc}

\usepackage{amsmath,amssymb}
\usepackage{fancybox}
\usepackage{epic}
\usepackage{eepic}
\usepackage{epsf}
\usepackage{epsfig}

\usepackage[vlined,ruled,linesnumbered]{algorithm2e}
\usepackage{subfigure,verbatim}

\newtheorem{fact}{Fact}

\newcommand{\R}{\ensuremath{\mathbb{R}}} 

\newcommand{\stackVC}{{\sc StackVC}}
\newcommand{\stackSP}{{\sc StackSP}}

\newcommand{\stackPrice}{{\sc Stack}}
\newcommand{\cov}{\mathcal{C}}

\begin{document}

\title[Stackelberg Network Pricing Games]{Stackelberg Network Pricing Games}

\author[lab1]{P. Briest}{Patrick Briest}
\address[lab1]{Department of Computer Science, The University of Liverpool, United Kingdom.
  \newline Supported by DFG grant Kr 2332/1-2 within Emmy Noether program.}
\email{{patrick.briest,p.krysta}@liverpool.ac.uk}  

\author[lab2]{M. Hoefer}{Martin Hoefer}
\address[lab2]{Department of Computer Science, RWTH Aachen University, Germany.
  \newline Supported by DFG Graduiertenkolleg 1298 ``AlgoSyn''.}	
\email{mhoefer@informatik.rwth-aachen.de}  

\author[lab1]{P. Krysta}{Piotr Krysta}



\begin{abstract}
  We study a multi-player one-round game termed Stackelberg Network
  Pricing Game, in which a \emph{leader} can set prices for a subset
  of $m$ priceable edges in a graph. The other edges have a fixed
  cost. Based on the leader's decision one or more \emph{followers}
  optimize a polynomial-time solvable combinatorial minimization
  problem and choose a minimum cost solution satisfying their
  requirements based on the fixed costs and the leader's prices. The
  leader receives as revenue the total amount of prices paid by the
  followers for priceable edges in their solutions, and the problem is
  to find revenue maximizing prices. Our model extends several known
  pricing problems, including single-minded and unit-demand pricing,
  as well as Stackelberg pricing for certain follower problems like
  shortest path or minimum spanning tree. Our first main result is a
  tight analysis of a single-price algorithm for the single follower
  game, which provides a $(1+\varepsilon) \log m$-approximation for
  any $\varepsilon >0$. This can be extended to provide a
  $(1+\varepsilon )(\log k + \log m)$-approximation for the general
  problem and $k$ followers. The latter result is essentially best
  possible, as the problem is shown to be hard to approximate within
  $\mathcal{O}(\log^\varepsilon k + \log^\varepsilon m)$. If followers
  have demands, the single-price algorithm provides a $(1+\varepsilon
  )m^2$-approximation, and the problem is hard to approximate within
  $\mathcal{O}(m^\varepsilon)$ for some $\varepsilon >0$. Our second
  main result is a polynomial time algorithm for revenue maximization
  in the special case of Stackelberg bipartite vertex cover, which is
  based on non-trivial max-flow and LP-duality techniques. Our results
  can be extended to provide constant-factor approximations for any
  constant number of followers.
\end{abstract}

\maketitle

\stacsheading{2008}{133-142}{Bordeaux}
\firstpageno{133}


\section{Introduction}
\label{intro}
Algorithmic pricing problems model the task of assigning revenue
maximizing prices to a retailer's set of products given some estimate
of the potential customers' preferences in purely
computational~\cite{Glynn02}, as well as strategic~\cite{Balcan+05}
settings. Previous work in this area has mostly focused on settings in
which these preferences are rather restricted, in the sense that
products are either {\em pure
  complements}~\cite{Balcan06,BK06,Guruswami05,HartlineKoltun} and
every customer is interested in exactly one subset of products or {\em
  pure
  substitutes}~\cite{Aggarwal04,BK07,Chawla+07,Glynn02,Guruswami05,HartlineKoltun},
in which case each customer seeks to buy only a single product out of
some set of alternatives. A customer's real preferences, however, are
often significantly more complicated than that and therefore pose some
additional challenges.

The modelling of consumer preferences has received considerable
attention in the context of {\em algorithmic mechanism
  design}~\cite{Nisan99} and {\em combinatorial
  auctions}~\cite{Cramton06}. The established models range from
relatively simple bidding languages to bidders that are represented by
oracles allowing certain types of queries, e.g., revealing the desired
bundle of items given some fixed set of prices. The latter would be a
somewhat problematic assumption in the theory of pricing algorithms,
where we usually assume to have access to a rather large number of
potential customers through some sort of sampling procedure and, thus,
are interested in preferences that allow for a compact kind of
representation.

In this paper we focus on customers that have non-trivial preferences,
yet can be fully described by their {\em types} and {\em budgets} and
do not require any kind of oracles. Assume that a company owns a
subset of the links in a given network. The remaining edges are owned
by other companies and have fixed publicly known prices and some
customer needs to purchase a path between two terminals in the
network. Since she is acting rational, she is going to buy the
shortest path connecting her terminals. How should we set the prices
on the priceable edges in order to maximize the company's revenue?
What if there is another customer, who needs to purchase, e.g., a
minimum cost spanning tree?

This type of pricing problem, in which preferences are implicitly
defined in terms of some optimization problem, is usually referred to
as {\em Stackelberg pricing}~\cite{Stackelberg34}. In the standard
2-player form we are given a {\em leader} setting the prices on a
subset of the network and a {\em follower} seeking to purchase a
min-cost network satisfying her requirements. We proceed by formally
defining the model before stating our results.

\subsection{Model and Notation}
In this paper we consider the following class of multi-player
one-round games. Let $G=(V,E)$ be a multi-graph. There are two types
of players in the game, one \emph{leader} and one or more
\emph{followers}. We consider two classes of \emph{edge} and
\emph{vertex games}, in which either the edges or the vertices have
costs. For most of the paper, we will consider edge games, but the
definitions and results for vertex games follow analogously. In an
edge game, the edge set $E$ is partitioned into two sets $E = E_p \cup
E_f$ with $E_p \cap E_f = \emptyset$. For each \emph{fixed-price} edge
$e \in E_f$ there is a fixed cost $c(e) \ge 0$. For each
\emph{priceable} edge $e \in E_p$ the leader can specify a price $p(e)
\ge 0$. We denote the number of priceable edges by $m = |E_p|$. Each
follower $i=1,\ldots,k$ has a set $\mathcal{S}_i \subset 2^E$ of
\emph{feasible subnetworks}. The \emph{weight} $w(S)$ of a subnetwork
$S \in \mathcal{S}_i$ is given by the costs of fixed-price edges and
the price of priceable edges,
\[ w(S) = \sum_{e \in S \cap E_f} c(e) + \sum_{e \in S \cap E_p}
p(e). \] 
The \emph{revenue} $r(S)$ of the leader from subnetwork $S$ is given
by the prices of the priceable edges that are included in $S$, i.e.,
\[ r(S) = \sum_{e \in S \cap E_p} p(e).\]
Throughout the paper we assume that for any price function $p$ every
follower $i$ can in polynomial time find a subnetwork
$S_i^*(p)$ of minimum weight. Our interest is to find the pricing
function $p^*$ for the leader that generates maximum revenue,
i.e.,
\[ p^* = \arg\max_p \sum_{i=1}^k r(S_i^*(p)).\]
We denote the value of this maximum revenue by $r^*$. To guarantee
that the revenue is bounded and the optimization problem is
non-trivial, we assume that there is at least one feasible subnetwork
for each follower $i$ that is composed only of fixed-price edges. In
order to avoid technicalities, we assume w.l.o.g. that among
subnetworks of identical weight the follower always chooses the one
with higher revenue for the leader. It is not difficult to see that in
the 2-player case we also need followers with a large number of
feasible subnetworks in order to make the problem interesting.

\begin{proposition}
\label{prop}
Given follower $j$ and a fixed subnetwork $S_j\in \mathcal{S}_j$, we can compute prices $p$ with $w(S_j)=\min _{S\in \mathcal{S}_j} w(S)$ maximizing $r(S_j)$ or decide that such prices do not exist in polynomial time. In the 2-player game, if $|\mathcal{S}|=\mathcal{O}(poly(m))$, revenue maximization can be done in polynomial time.
\end{proposition}

The proof of Proposition~\ref{prop} will appear in the full
version. In general we will refer to the revenue optimization problem
by \stackPrice. Note that our model extends the previously considered
pricing models and is essentially equivalent to pricing with general
valuation functions, a problem that has independently been considered
in~\cite{BalcanTR}. Every general valuation function can be expressed
in terms of Stackelberg network pricing on graphs, and our algorithmic
results apply in this setting as well.

\subsection{Previous Work and New Results}
\label{results}
The single-follower shortest path Stackelberg pricing problem ({\sc
  StackSP}) has first been considered by Labb\'{e} et
al.~\cite{Labbe98}, who derive a bilevel LP formulation of the problem
and prove NP-hardness. Roch et al.~\cite{Roch05} present a first
polynomial time approximation algorithm with a provable performance
guarantee, which yields logarithmic approximation ratios. Bouhtou et
al.~\cite{Bouhtou+} extend the problem to multiple (weighted)
followers and present algorithms for a restricted shortest path
problem on parallel links. For an overview of most of the initial work
on Stackelberg network pricing the reader is referred
to~\cite{Hoesel06}. A different line of research has been
investigating the application of Stackelberg pricing to network
congestion games in order to obtain low congestion Nash equilibria for
sets of selfish followers~\cite{Cole03,Roughgarden04,Swamy07}.

More recently, Cardinal et al.~\cite{Cardinal07} initiated the
investigation of the corresponding minimum spanning tree ({\sc
  StackMST}) game, again obtaining a logarithmic approximation
guarantee and proving APX-hardness. Their {\em single-price
  algorithm}, which assigns the same price to all priceable edges,
turns out to be even more widely applicable and yields similar
approximation guarantees for any matroid based Stackelberg game.

The first result of our paper is a generalization of this result to
general Stackelberg games. The previous limitation to matroids stems
from the difficulty to determine the necessarily polynomial number of
candidate prices that can be tested by the algorithm. We develop a
novel characterization of the small set of {\em threshold prices} that
need to be tested and obtain a polynomial time
$(1+\varepsilon)H_m$-approximation (where $H_m$ denotes the $m$'th
harmonic number) for arbitrary $\varepsilon >0$, which turns out to be
perfectly tight for shortest path as well as minimum spanning tree
games. This result is found in Section~\ref{single-price}.

We then extend the analysis to multiple followers, in which case the
approximation ratio becomes $(1+\varepsilon)(H_k+H_m)$. This can be
shown to be essentially best possible by an approximation preserving
reduction from single-minded combinatorial
pricing~\cite{Demaine06+}. Extending the problem even further, we also
look at the case of multiple {\em weighted} followers, which arises
naturally in network settings where different followers come with
different routing demands. It has been conjectured before that no
approximation essentially better than the number of followers is
possible in this scenario. We disprove this conjecture by presenting
an alternative analysis of the single-price algorithm resulting in an
approximation ratio of $(1+\varepsilon )m^2$. Additionally, we derive
a lower bound of $\mathcal{O}(m^{\varepsilon})$ for the weighted
player case. This resolves a previously open problem
from~\cite{Bouhtou+}. The results on multiple followers are found in
Section~\ref{multiFollowers}.

The generic reduction from single-minded to Stackelberg pricing yields
a class of networks in which we can price the vertices on one side of
a bipartite graph and players aim to purchase minimum cost vertex
covers for their sets of edges. This motivates us to return to the
classical Stackelberg setting and consider the 2-player bipartite
vertex cover game ({\sc StackVC}). As it turns out, this variation of
the game allows polynomial-time algorithms for exact revenue
maximization using non-trivial algorithmic techniques. We first
present an upper bound on the possible revenue in terms of the
min-cost vertex cover not using any priceable vertices and the minimum
portion of fixed cost in any possible cover. Using iterated max-flow
computations, we then determine a pricing with total revenue that
eventually coincides with our upper bound. These results are found in
Section~\ref{StackVC}.

Finally, Section~\ref{conclusions} concludes and presents several
intriguing open problems for further research. Some of the proofs have
been omitted due to space limitations.

\section{A Single-Price Algorithm for a Single Follower}
\label{single-price}
Let us assume that we are faced with a single follower and let $c_0$
denote the cost of a cheapest feasible subnetwork for the follower not
containing any of the priceable edges. Clearly, we can compute $c_0$
by assigning price $+\infty$ to all priceable edges and simulating the
follower on the resulting network. The {\em single-price algorithm}
proceeds as follows. For $j=0,\ldots ,\lceil \log c_0\rceil$ it
assigns price $p_j=(1+\varepsilon )^j$ to all priceable edges and
determines the resulting revenue $r(p_j)$. It then simply returns the
pricing that results in maximum revenue. We present a logarithmic
bound on the approximation guarantee of the single-price algorithm.

\begin{theorem}
\label{t:single-price}
Given any $\varepsilon >0$, the single-price algorithm computes an
$(1+\varepsilon)H_m$-approximation with respect to $r^*$, the revenue
of an optimal pricing. 
\end{theorem}

\subsection{Analysis}
The single-price algorithm has previously been applied to a number of
different combinatorial pricing
problems~\cite{Aggarwal04,Guruswami05}. The main issue in analyzing
its performance guarantee for Stackelberg pricing is to determine the
right set of candidate prices. We first derive a precise
characterization of these candidates and then argue that the geometric
sequence of prices tested by the algorithm is a good enough
approximation.
Slightly abusing notation, we let $p$ refer to both price $p$ and the
assignment of this price to all priceable edges. If there exists a
feasible subnetwork for the follower that uses at least $j$ priceable
edges, we let
\[ \theta_j=\max \Bigl\{ p\, \Bigr| \Bigl. \, |S^{\star}(p)\cap E_p|\ge j\Bigr\}\]
be the largest price at which such a subnetwork is chosen. If no
feasible subnetwork with at least $j$ priceable edges exists, we set
$\theta_j=0$. As we shall see, these thresholds are the key to prove
Theorem~\ref{t:single-price}. 

We want to derive an alternative characterization of the
values of $\theta_j$. For each $1 \le j\le m$ we let $c_j$ refer to
the minimum sum of prices of fixed-price edges in any feasible
subnetwork containing at most $j$ priceable edges, formally
\[ c_j=\min \Bigl\{\sum _{e\in S\cap E_f}f_e\, \Bigr| \Bigl. \, S\in \mathcal{S}\,
:\, |S\cap E_p|\le j\Bigr\},\] 
and $\Delta_j=c_0-c_j$. For ease of notation let $\Delta
_0=0$. Consider the set of points $(0,\Delta _0)$, $(1,\Delta _1),\ldots
,(m,\Delta _m)$ on the plane. By $\mathcal{H}$ we refer to a minimum
selection of points spanning the upper convex hull of the point
set. It is a straightforward geometric observation that we can define
$\mathcal{H}$ as follows:

\begin{fact}
\label{fact:hul}
Point $(j,\Delta _j)$ belongs to $\mathcal{H}$ if and only if
%
$\min_{i<j} \frac{\Delta_j-\Delta_i}{j-i}> \max _{j<k}
\frac{\Delta_k-\Delta_j}{k-j}.$
%
\end{fact} 

We now return to the candidate prices. By definition we have that
$\theta_1 \ge \theta_2 \ge \cdots \ge \theta_m$. We say that
$\theta_j$ is {\em true threshold value} if $\theta _j>\theta _{j+1}$,
i.e., if at price $\theta_j$ the subnetwork chosen by the follower
contains exactly $j$ priceable edges. Let $i_1 < i_2 < \cdots
<i_{\ell}$ denote the indices, such that $\theta_{i_k}$ are true
threshold values and for ease of notation define $i_0=0$. For an
example, see Figure~\ref{fig:convexHull}.

\begin{lemma}
\label{t:hull}
$\theta_j$ is true threshold value if and only if $(j,\Delta_j)$
belongs to $\mathcal{H}$.
\end{lemma}
\begin{proof}
"$\Rightarrow$" Let $\theta_j$ be true threshold value, i.e., at
price $\theta_j$ the chosen subnetwork contains exactly $j$ priceable
edges. We observe that at any price $p$ the cheapest subnetwork
containing $j$ priceable edges has cost $c_j + j \cdot p = c_0 -
\Delta_j + j \cdot p$. Thus, at price $\theta_j$ it must be the case
that $\Delta_j - j\cdot \theta_j \ge \Delta_i - i\cdot \theta_j$ for
all $i<j$ and $\Delta_j - j\cdot \theta_j > \Delta_k - k \cdot
\theta_j$ for all $j<k$. It follows that
\[ \min_{i<j} \frac{\Delta_j-\Delta_i}{j-i}\ge \theta_j> \max_{j<k}
\frac{\Delta_k-\Delta_j}{k-j},\] 
and, thus, we have that $(j,\Delta_j)$ belongs to $\mathcal{H}$.

"$\Leftarrow$" Assume now that $(j,\Delta_j)$ belongs to $\mathcal{H}$
and let
\[ p=\min_{i<j} \frac{\Delta_j-\Delta_i}{j-i}.\]
Consider any $k<j$. It follows that
%
$\Delta_k - k \cdot p=\Delta_j - j \cdot p-(\Delta_j -
\Delta_k)+(j-k)p \le \Delta_j-j\cdot p,$
%
since $p \le (\Delta_j-\Delta_k)/(j-k)$ and, thus, the network chosen
at price $p$ cannot contain less than $j$ priceable edges. Analogously,
let $k>j$. Using $p>(\Delta_k-\Delta_j)/(k-j)$ we obtain
%
$\Delta_k - k \cdot p = \Delta_j - j \cdot p + (\Delta_k -
\Delta_j)-(k-j)p < \Delta_j - j \cdot p,$
%
and, thus, the subnetwork chosen at price $p$ contains exactly $j$
priceable edges. We conclude that $\theta_j$ is a true threshold.
\end{proof}

\begin{figure}
\centering
\epsfig{file=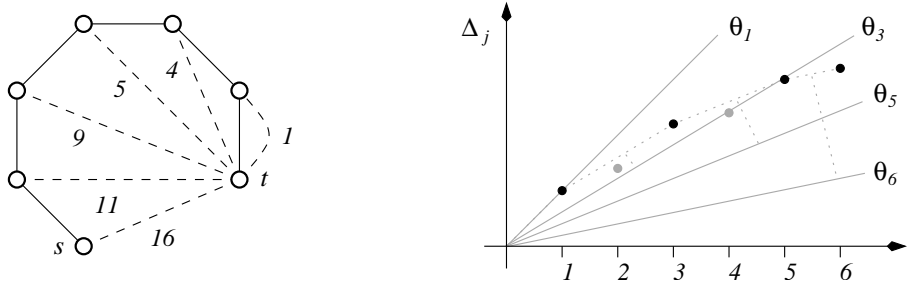,width=12cm}
\caption{\label{fig:convexHull} A geometric interpretation of (true)
  threshold values $\theta_j$. The follower seeks to purchase a
  shortest path from $s$ to $t$, dashed edges are fixed-cost.}
\end{figure}
 
It is not difficult to see that the price $p$ defined in the second
part of the proof of Lemma~\ref{t:hull} is precisely the threshold
value $\theta_j$. Let $\theta_{i_k}$ be any true threshold. Since
points $(i_0,\Delta_{i_0}),\ldots ,(i_{\ell},\Delta_{i_{\ell}})$
define the convex hull we can write that $\min_{i<i_k} (\Delta_{i_k} -
\Delta_i)/(i_k-i) = (\Delta_{i_k} -\Delta_{i_{k-1}})/(i_k-i_{k-1})$.
%
%
We state this important fact again in the following lemma.

\begin{lemma}
\label{t:thresholds}
For all $1 \le k \le \ell$ it holds that
$\theta_{i_k}=\frac{\Delta_{i_k} - \Delta_{i_{k-1}}}{i_k-i_{k-1}}$.
\end{lemma}

From the fact that points $(i_0,\Delta_{i_0}),\ldots
,(i_{\ell},\Delta_{i_{\ell}})$ define the convex hull we know that
$\Delta_{i_{\ell}} = \Delta_m$, i.e., $\Delta_{i_{\ell}}$ is the
largest of all $\Delta$-values. On the other hand, each $\Delta_j$
describes the maximum revenue that can be made from a subnetwork with
at most $j$ priceable edges and, thus, $\Delta_m$ is clearly an upper
bound on the revenue made by an optimal price assignment. 

\begin{fact}
\label{fact:opt}
It holds that $r^* \le \Delta_{i_{\ell}}$.
\end{fact}

By definition of the $\theta_j$'s it is clear that at any price below
$\theta_{i_k}$ the subnetwork chosen by the follower contains no less
than $i_k$ priceable edges. Furthermore, for each $\theta_{i_k}$ the
single-price algorithm tests a candidate price that is at most a
factor $(1+\varepsilon)$ smaller than $\theta_{i_k}$. Let
$r(p_{i_k})$, $r(\theta_{i_k})$ denote the revenue that results from
assigning price $p_{i_k}$ or $\theta_{i_k}$ to all priceable edges,
respectively. 

\begin{fact}
\label{fact:prices}
For each $\theta_{i_k}$ there exists a price $p_{i_k}$ with
$(1+\varepsilon)^{-1}\theta_{i_k}\le p_{i_k}\le \theta_{i_k}$ that
is tested by the single-price algorithm. Especially, it holds that
$r(p_{i_k})\ge (1+\varepsilon)^{-1}r(\theta_{i_k})$ 
\end{fact}

Finally, we know that the revenue made by assigning price
$\theta_{i_k}$ to all priceable edges is $r(\theta_{i_k})=i_k\cdot
\theta_{i_k}$. Let $r$ denote the revenue of the single-price
solution returned by the algorithm. We have: 
\begin{eqnarray*}
(1+\varepsilon)\cdot H_m\cdot r & = & (1+\varepsilon)\sum_{j=1}^m
\frac{r}{j} \ge (1+\varepsilon) \sum_{k=1}^{\ell}
\sum_{j=i_{k-1}+1}^{i_k}\frac{r}{j} \ge (1+\varepsilon)\sum_{k=1}^{\ell}
\sum_{j=i_{k-1}+1}^{i_k}\frac{r(p_{i_k})}{j}\\ 
 & \ge & \sum_{k=1}^{\ell}
 \sum_{j=i_{k-1}+1}^{i_k}\frac{r(\theta_{i_k})}{j}\ge
 \sum_{k=1}^{\ell} \sum_{j=i_{k-1}+1}^{i_k}\frac{i_k\cdot
   \theta_{i_k}}{j}\\ 
 & \ge & \sum_{k=1}^{\ell} (i_k-i_{k-1})\frac{i_k\cdot
   \theta_{i_k}}{i_k}= \sum_{k=1}^{\ell} (\Delta_{i_k}
 - \Delta_{i_{k-1}}) \mbox{ , by Lemma \ref{t:thresholds}}\\ 
 & = & \Delta_{i_{\ell}} - \Delta_0=\Delta_{i_{\ell}} \ge r^*.
\end{eqnarray*}
This concludes the proof of Theorem \ref{t:single-price}.


\begin{figure}
\centering
\includegraphics[scale=0.4]{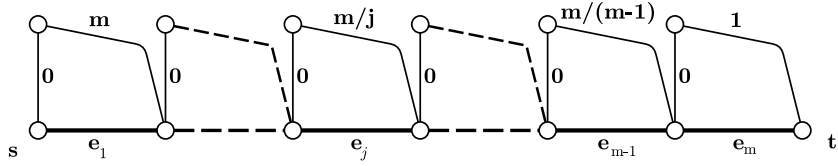}
\caption{\label{fig:lowerBound} An instance of Stackelberg Shortest
  Path, on which the analysis of the approximation guarantee of the
  single-price algorithm is tight. Bold edges are priceable, vertex
  labels of regular edges indicate cost. The instance yields tightness
  of the analysis also for Stackelberg Minimum Spanning Tree.}
\end{figure}

\subsection{Tightness}
The example in Figure~\ref{fig:lowerBound} shows that our analysis of
the single-price algorithm's approximation guarantee is tight. The
follower wants to buy a path connecting vertices $s$ and $t$. In an
optimal solution we set the price of edge $e_j$ to $m/j$. Then edges
$e_1,\ldots, e_m$ form a shortest path of cost $mH_m$. On the other
hand, assume that all edges $e_1,\ldots, e_m$ are assigned the same
price $p$. Every choice will lead to a revenue of at most $m$.
Similar 
results apply if the follower purchases
a minimum spanning tree instead of a shortest path.

The best known lower bound for 2-player Stackelberg pricing is found
in~\cite{Cardinal07}, where APX-hardness is shown for the minimum
spanning tree case. To the authors' best knowledge, up to now no
non-constant inapproximability results have been proven. We proceed by
extending our results to multiple followers, in which case previous
results on other combinatorial pricing problems yield strong lower
bounds.

\section{Extension to Multiple Followers}
\label{multiFollowers}
In this section we extend our results on general Stackelberg network
pricing to scenarios with multiple followers. Recall that each
follower $j$ is characterized by her own collection $\mathcal{S}_j$ of
feasible subnetworks and $k$ denotes the number of
followers. Section~\ref{logk+logm} extends the analysis from the
single follower case to prove a tight bound of
$(1+\varepsilon)(H_k+H_m)$ on the approximation guarantee of the
single-price algorithm. 
In addition, it presents an alternative analysis that applies even in
the case of weighted followers and yields approximation guarantees
that do not depend on the number of
followers. Section~\ref{lowerBounds} derives (near) tight
inapproximability results based on known hardness results for
combinatorial pricing. Proofs are omitted due to space limitations.


\subsection{Guarantees of the Single-Price Algorithm}
\label{logk+logm}
Let an instance of Stackelberg network pricing with some number $k\ge
1$ of followers be given. We 
obtain a similar bound on the single-price algorithm's approximation
guarantee.
\begin{theorem}
\label{t:single-price-multi-1}
The single-price algorithm computes an
$(1+\varepsilon)(H_k+H_m)$-approximation with respect to $r^*$, the
revenue of an optimal pricing, for {\sc Stack} with multiple followers.
\end{theorem}
The proof of Theorem~\ref{t:single-price-multi-1} reduces the problem
to the single player case. However, it relies essentially on the fact
that we are considering the single-price algorithm. 
It does not imply anything about the relation of these two cases in
general.

\label{m2}
An even more general variation of Stackelberg pricing, in which we
allow multiple {\em weighted} followers, arises naturally in the
context of network pricing games with different demands for each
player. This model has been previously considered in~\cite{Bouhtou+}.
Formally, for each follower $j$ we are given her {\em demand} $d_j \in
\R_0^+$. Given followers buying subnetworks $S_1,\ldots,S_k$, the
leader's revenue is defined as
$ \sum_{j=1}^k d_j\sum_{e\in S_j\cap E_p}p(e).$
It has been conjectured before that in the weighted case no
approximation guarantee essentially beyond $\mathcal{O}(k\cdot \log
m)$ is possible \cite{Roch05}. We show that an alternative analysis of
the single-price algorithm yields ratios that do not depend on the
number of followers. 

\begin{theorem}
\label{t:single-price-multi-2}
The single-price algorithm computes an $(1+\varepsilon
)m^2$-approximation with respect to $r^*$, the revenue of an optimal
pricing, for {\sc Stack} with multiple weighted followers. 
\end{theorem}

\subsection{Lower Bounds}
\label{lowerBounds}
Hardness of approximation of Stackelberg pricing with multiple
followers follows immediately from known results about other
combinatorial pricing
models. 
Theorem~\ref{t:hardness1} is based on a reduction from the (weighted)
unit-demand envy-free pricing problem with uniform budgets, which is
known to be inapproximable within $\mathcal{O} (m^{\varepsilon})$ ($m$
denotes the number of products)~\cite{B06}. Here we are given a
universe of products and a collection of (weighted) customers, each of
which buys the cheapest product out of some set of alternatives with a
price not exceeding her budget. The resulting Stackelberg game is an
instance of the so-called {\em river tariffication problem}. Each
player needs to route her demand along one out of a number of parallel
links connecting her respective source and sink pair. One direct fixed
price connection determines her maximum budget for purchasing a
priceable link. Theorem~\ref{t:hardness1} resolves an open problem
from~\cite{Bouhtou+}. The construction is depicted in
Figure~\ref{fig:hardness1}.

\begin{theorem}
\label{t:hardness1}
The Stackelberg network pricing problem with multiple weighted
followers is hard to approximate within $\mathcal{O}(m^{\varepsilon})$
for some $\varepsilon >0$, unless NP $\subseteq \bigcap_{\delta >0}$
BPTIME($2^{n^{\delta}}$). The same holds for the river tariffication
problem. 
\end{theorem}

\begin{figure}
\centering
\subfigure[]{
\label{fig:river} 
\label{fig:hardness1}
\centering
\includegraphics[scale=0.3]{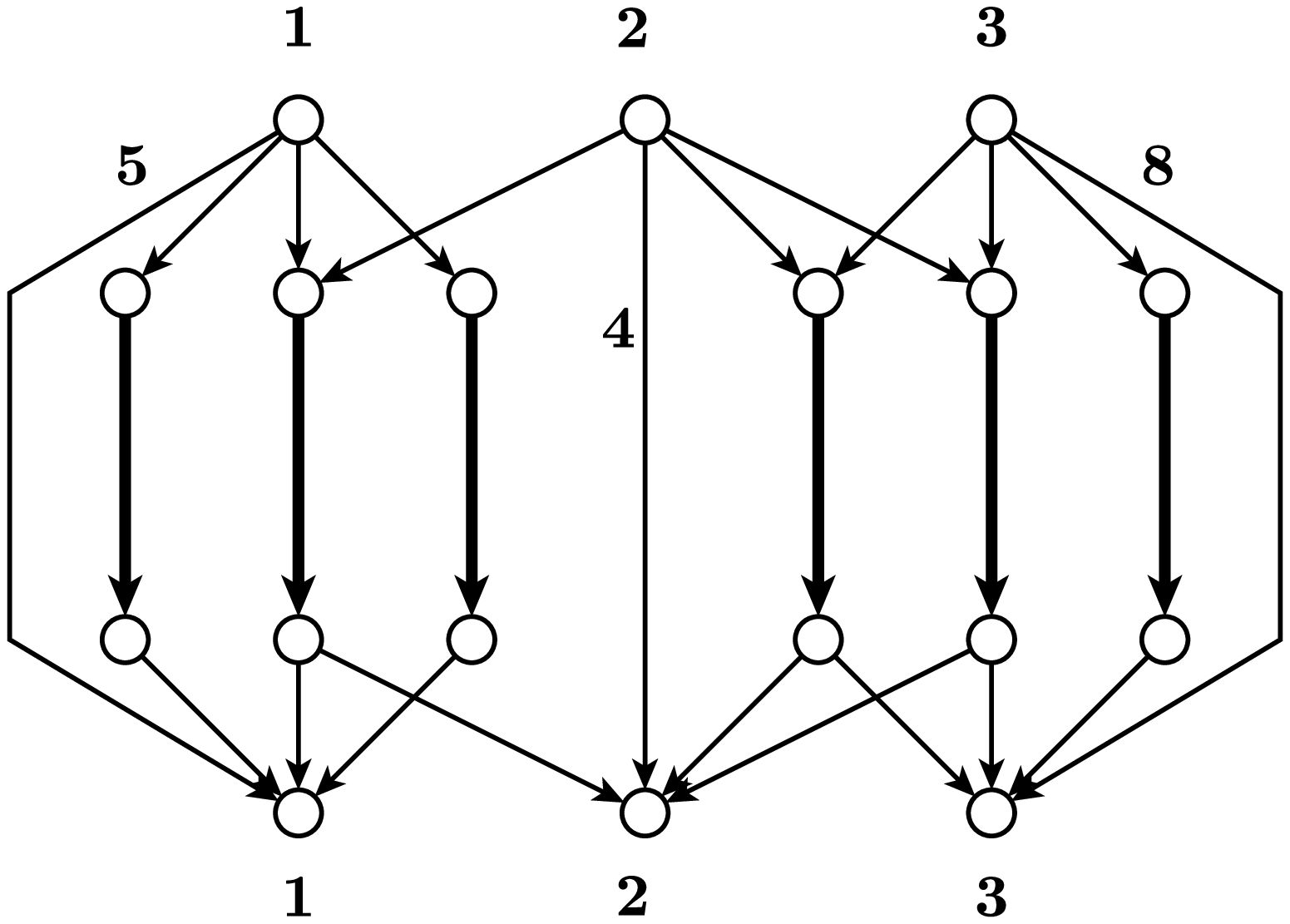}}
\hspace{2cm}
\subfigure[]{
\label{fig:VCsingle} 
\label{fig:hardness2}
\centering
\includegraphics[scale=0.3]{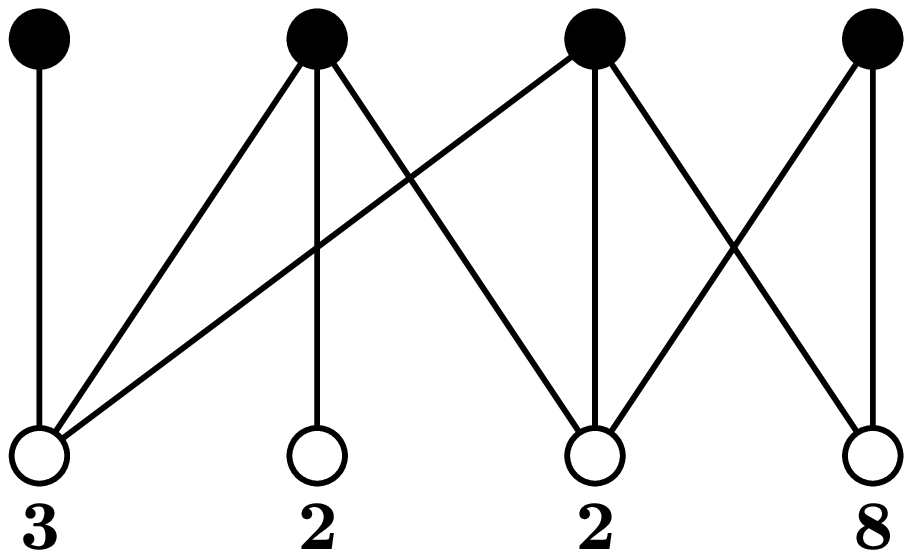}}

\caption{Reductions from pricing problems to Stackelberg pricing. (a)
  Unit-demand reduces to directed \stackSP. Bold edges are priceable,
  edge labels indicate cost. Regular edges without labels have cost
  0. Vertex labels indicate source-sink pairs for the followers. (b)
  Single-minded pricing reduces to bipartite \stackVC. Filled vertices
  are priceable, vertex labels indicate cost. For each customer there
  is one follower, who strives to cover all incident edges.}
\end{figure}

Theorem~\ref{t:hardness2} is based on a reduction from the
single-minded combinatorial pricing problem, in which each customer is
interested in a subset of products and purchases the whole set if the
sum of prices does not exceed her budget. Single-minded pricing is
hard to approximate within $\mathcal{O}(\log ^{\varepsilon}k+\log
^{\varepsilon}m)$~\cite{Demaine06+}, where $k$ and $m$ denote the
numbers of customers and products,
respectively. Theorem~\ref{t:hardness2} shows that the single-price
algorithm is essentially best possible for multiple unweighted
followers.

\begin{theorem}
\label{t:hardness2}
The Stackelberg network pricing problem with multiple unweighted
followers is hard to approximate within $\mathcal{O}(\log
^{\varepsilon}k+\log ^{\varepsilon}m)$ for some $\varepsilon >0$,
unless NP $\subseteq \bigcap _{\delta >0}$
BPTIME($2^{n^{\delta}}$). The same holds for bipartite Stackelberg
Vertex Cover Pricing ({\sc StackVC}).
\end{theorem}

The idea for the proof of Theorem~\ref{t:hardness2} is illustrated in
Figure~\ref{fig:hardness2}. We define an instance of \stackVC\ in
bipartite graphs. Vertices on one side of the bipartition are
priceable and represent the universe of products, vertices on the
other side encode customers and have fixed prices corresponding to the
respective budgets. For each customer we define a follower in the
Stackelberg game with edges connecting the customer vertex and all
product vertices the customer wishes to purchase. Now every follower
seeks to buy a min-cost vertex cover for her set of edges. We proceed
by taking a closer look at this special type of Stackelberg pricing
game and especially focus on the interesting case of a single
follower.


\vskip-0.3cm
\section{Stackelberg Vertex Cover}
\label{StackVC}

Stackelberg Vertex Cover Pricing is a vertex game, however, the
approximation results for the single-price algorithm continue to
hold. Note that in general the vertex cover problem is hard, hence we
focus on settings, in which the problem can be solved in polynomial
time. In bipartite graphs the problem can be solved optimally by using
a classic and fundamental max-flow/min-cut argumentation. If all
priceable vertices are in one side of the partition, then for multiple
followers there is evidence that the single-price algorithm is
essentially best possible. Our main theorem in this section states 
that the setting with a single follower can be solved exactly. As a
consequence, general bipartite \stackVC\ can be approximated by a
factor of 2.

\begin{theorem}
\label{theo:bip:one_side}
If for a bipartite graph $G=(A \cup B, E)$ we have $V_p \subseteq A$,
then there is a polynomial time algorithm computing an optimal price
function $p^*$ for \stackVC. 
\end{theorem}

Before we prove the theorem, we mention that the standard problem of
minimum vertex cover in a bipartite graph $G$ with disjoint vertex
sets $A$, $B$ and edges $E \subseteq A \times B$ can be solved by the
following application of LP-duality. The LP-dual is interpreted as a
maximum flow problem on an adjusted flow network $G_d$. In particular,
$G_d$ is constructed by adding a source $s$ and a sink $t$ to $G$ and
connecting $s$ to all vertices $v \in A$ with directed edges $(s,v)$,
and $t$ to all vertices $v \in B$ with directed edges $(v,t)$. Each
such edge gets as capacity the cost of the involved original vertex -
i.e. $p(v)$ for $v \in V_p$ or $c(v)$ if $v \in V_f$. Furthermore, all
original edges of the graph are directed from $A$ to $B$ and their
capacity is set to infinity. The value of a maximum $s$-$t$-flow
equals the cost of a minimum cut, and in addition the cost of a
minimum cost vertex cover of the graph $G$ (for an example see
Figure~\ref{fig:VCopt}). To obtain such a cover consider an
\emph{augmenting $s$-$t$-path} in $G_d$, which is a path traversing
only forward edges with slack capacity and backward edges with
non-zero flow. The maximum flow can be computed by iteratively
increasing flow along such paths. The vertices in the minimum vertex
cover then correspond to incident edges in a minimum cut. In
particular, the minimum vertex cover includes a vertex $v \in A$ if
the flow allows no augmenting $s$-$v$-path from $s$ to $v$, i.e. if
every path from $s$ to $v$ has at least one backward edge with no
flow, or at least one forward edge without slack capacity.

We use a similar idea to obtain the optimal pricing for \stackVC. Let
$n=|V_p|$ and the values $c_j$ for $1 \le j\le n$ denote the minimum
sum of prices of fixed-price vertices in any feasible subnetwork
containing at most $j$ priceable vertices. Then, $\Delta_j = c_0 -
c_j$ are again upper bounds on the revenue that can be extracted from
a network that includes at most $j$ priceable vertices. We thus have
$r^* \le \Delta_n$.

\begin{figure}
\centering
\subfigure[]{
\centering
\includegraphics[scale=0.3]{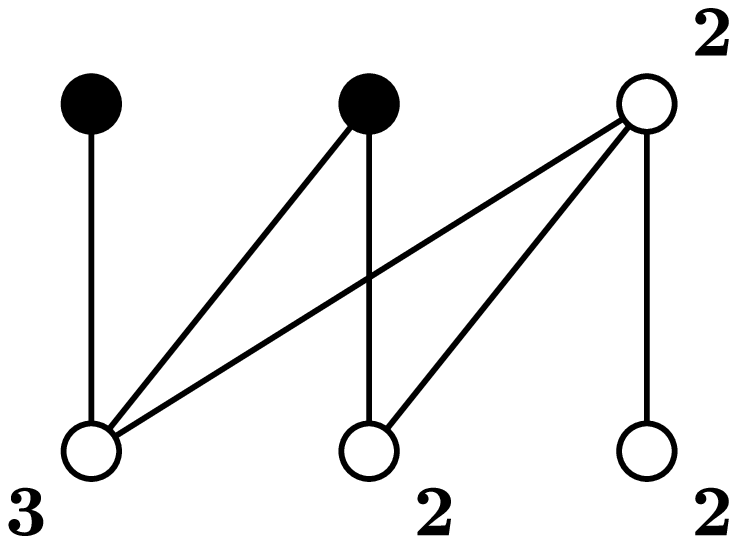}}
\hspace{1cm}
\subfigure[]{
\centering
\includegraphics[scale=0.3]{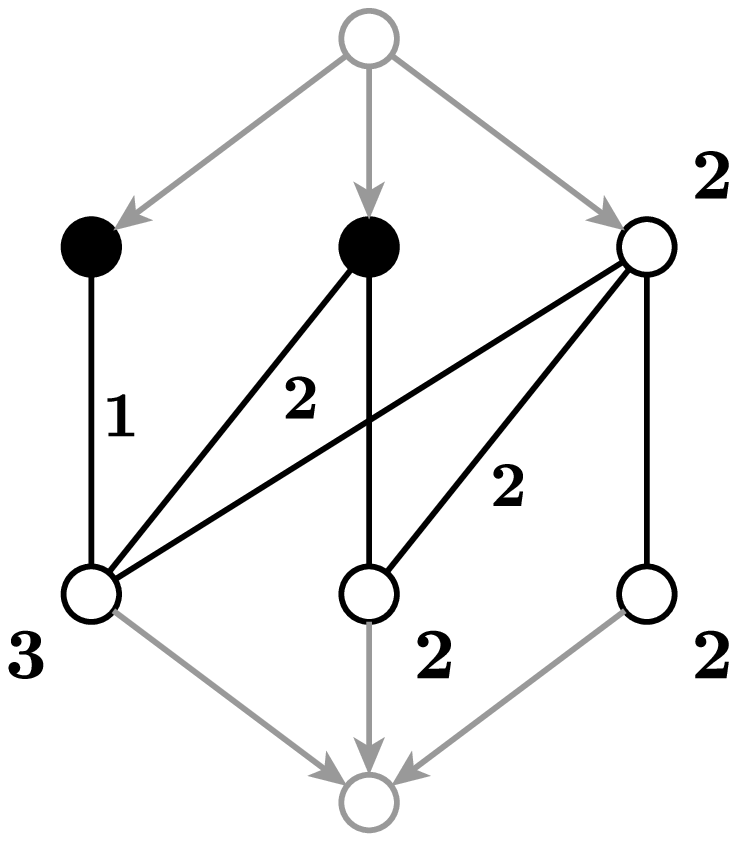}}
\hspace{1cm}
\subfigure[]{
\centering
\includegraphics[scale=0.3]{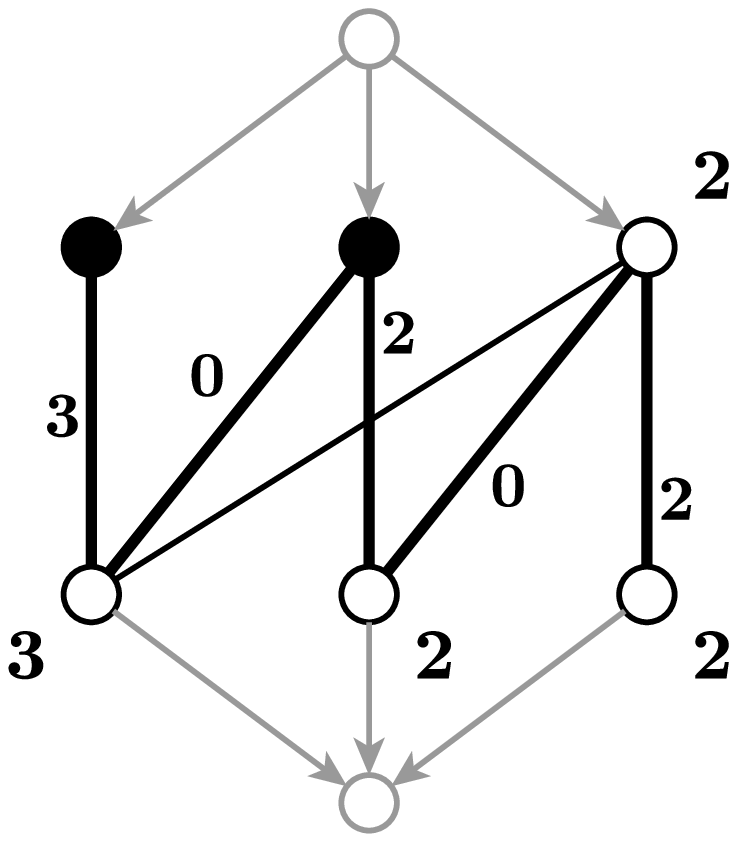}}
\caption{\label{fig:VCopt} Construction to solve bipartite \stackVC\
  with priceable vertices in one partition and a single
  follower. Filled vertices are priceable, vertex labels indicate
  cost. (a) A graph $G$; (b) The flow network $G_d$ obtained from
  $G$. Grey parts are source and sink added by the
  transformation. Edge labels indicate a suboptimal $s$-$t$-flow; (c)
  An augmenting path $P$ indicated by bold edges and the resulting
  flow. Every such path $P$ starts with a priceable vertex, and all
  priceable vertices remain in the optimum cover at all times.}
\end{figure}

\begin{algorithm}
\caption{\label{algo:bip:onepartition} Solving \stackVC\ in bipartite
  graphs with $V_p \subseteq A$}
\dontprintsemicolon
Construct the flow network $G_d$ by adding nodes $s$ and $t$\;
Set $p(v) = 0$ for all $v \in V_p$ \;
Compute a maximum $s$-$t$-flow $\phi$ in $G_d$ \;
\While{there is $v \in V_p$ s.t. increasing $p(v)$ yields an
  augmenting $s$-$t$-path $P$}{
Increase $p(v)$ and $\phi$ along $P$ as much as
 possible\;\label{step:aug}
}
\end{algorithm}

Suppose all priceable vertices are located in one partition $V_p
\subseteq A$ and consider Algorithm~\ref{algo:bip:onepartition}.  We
denote by $\cov_{ALG}$ the cover calculated by
Algorithm~\ref{algo:bip:onepartition}. At first, when computing the
maximum flow on $G_d$ holding all $p(v) = 0$, the algorithm obtains a
flow of $c_n$. We first note that in the following while-loop we will
never face a situation, in which there is an augmenting $s$-$t$-path
(traversing forward edges with slack capacity and backward edges with
non-zero flow) starting with a fixed-price vertex. We call such a path
a \emph{fixed} path, while an augmenting $s$-$t$-path starting with a
priceable vertex is called a \emph{price} path.
\begin{lemma}
\label{t:pricePath}
Every augmenting path considered in the while-loop of
Algorithm~\ref{algo:bip:onepartition} is a price path.
\end{lemma}
\begin{proof}
We prove the lemma by induction on the while-loop and by
contradiction. Suppose that in the beginning of the current iteration 
there is no fixed path. In particular, this is true for the first
iteration of the while-loop. Then, suppose that after we have
increased the flow over a price path $P_p$, a fixed path $P_f$ is
created. $P_f$ must include some of the edges of $P_p$. Consider the 
vertex $w$ at which $P_f$ hits $P_p$. By following $P_f$ from $s$ to
$w$ and $P_p$ from $w$ to $t$ there is a fixed path, which must have
been present before flow was increased on $P_p$. This is a
contradiction and proves the lemma. 
\end{proof}
Recall from above that the optimum cover contains a vertex $v \in A$
if there is no augmenting $s$-$v$-path from $s$ to $v$. In particular,
this means that for a vertex $v \in A \cap \cov$ the following two
properties are fulfilled: (1) there is no slack capacity on edge
$(s,v)$; (2) there is no augmenting $s$-$v$-path from $s$ over a
different vertex $v' \in A$. As the algorithm always adjusts the price
of a vertex $v$ to equal the current flow on $(s,v)$, 
only the violation of property~(2) can force a vertex $v
\in V_p$ to leave the cover. In particular, such an augmenting
$s$-$v$-path must start with a fixed-price vertex, and it must reach
$v$ by decreasing flow over one of the original edges $(v,w)$ for $w
\in B$. We call such a path a \emph{fixed $v$-path}.

\begin{lemma}
\label{t:noVpath}
Algorithm~\ref{algo:bip:onepartition} creates no fixed $v$-path for
any priceable vertex $v \in V_p$.
\end{lemma}

The proof of Lemma \ref{t:noVpath} is similar to the proof of Lemma
\ref{t:pricePath} and will appear in the full version. As there is no
augmenting path from $s$ to any priceable vertex at any time, the
following lemma is now obvious.
\begin{lemma}
$\cov_{ALG}$ includes all priceable vertices. 
\end{lemma}

{\em Proof of Theorem~\ref{theo:bip:one_side}.}\, Finally, we can
proceed to argue that the computed pricing is optimal. Suppose that
after executing Algorithm~\ref{algo:bip:onepartition} we increase
$p(v)$ over $\phi(s,v)$ for any priceable vertex $v$. As we are at the
end of the algorithm, it does not allow us to increase the flow in the
same way. Thus, the adjustment creates slack capacity on all the edges
$(s,v)$ for any $v \in V_p$ and causes every priceable vertex to leave
$\cov_{ALG}$. The new cover must be the cheapest cover that excludes
every priceable vertex, i.e. it must be $\cov_0$ and have cost
$c_0$. As we have not increased the flow, we know that the cost of
$\cov_{ALG}$ is also $c_0$. Note that before starting the while-loop
the cover was $\cov_n$ of cost $c_n$. As all flow increase in the
while-loop was made over price paths and all the priceable vertices
stay in the cover, the revenue of $\cov_{ALG}$ must be $c_0 - c_n =
\Delta_n$. This is an upper bound on the optimum revenue, and hence
the price function $p_{ALG}$ derived with the algorithm is
optimal. Finally, notice that adjusting the price of the priceable
vertices in each iteration is not necessary. We can start with
computing $\cov_n$ and for the remaining while-loop set all prices to
$+\infty$. This will result in the desired flow, which directly
generates the final price for every vertex $v$ as flow on
$(s,v)$. Hence, we can get optimal prices with an adjusted run of the
standard polynomial time algorithm for maximum flow in $G_d$. This
proves Theorem~\ref{theo:bip:one_side}.\qed
%
%
\begin{theorem}
\label{theo:bip:two_side}
There is a polynomial time $2$-approximation algorithm for bipartite
\stackVC.
\end{theorem}
In Theorem~\ref{theo:bip:two_side} we use the previous analysis to get
a $2$-approximation of the optimum revenue for general bipartite
\stackVC. This results in a $2k$-approximation for any number of $k$
followers. In contrast, the analysis of the single-price algorithm is
tight even for one follower and all priceable vertices in one
partition. Moreover, bipartite \stackVC\ for at least two followers is
NP-hard by a reduction from the highway pricing problem~\cite{BK06}.

\vskip-0.3cm
\section{Open problems}
\label{conclusions}

There are a number of important open problems that arise from our
work. We believe that the single-price algorithm is essentially best
possible even for a single follower and general Stackelberg pricing
games. However, there is no matching logarithmic lower bound, and the
best lower bound remains APX-hardness from~\cite{Cardinal07}. In
addition, we believe that for weighted followers a better upper bound
than $m^2$ is possible, which would decrease the gap to the
$\Omega(m^\varepsilon)$ lower bound we observed. More generally,
extending other algorithm design techniques to cope with pricing
problems is a major open problem.

\end{document}